\documentclass[letterpaper, 10 pt, conference]{ieeeconf}
  \IEEEoverridecommandlockouts                              
  \usepackage{cite}
  \usepackage{graphicx}
  \usepackage{psfrag}
  \usepackage{lipsum}
  \usepackage[utf8]{inputenc}
  \usepackage{commath}
  \usepackage{amsfonts}
  \usepackage{times}
  \usepackage{xcolor}
  \renewcommand{\baselinestretch}{0.99}
  \newtheorem{theorem}{Theorem}[section]
  
  \newtheorem{lemma}[theorem]{Lemma}
  \newtheorem{proposition}[theorem]{Proposition}
  \newtheorem{remark}[theorem]{Remark}

  \usepackage{breqn}
  \usepackage{mathtools}
  \DeclarePairedDelimiter\ceil{\lceil}{\rceil}
 \DeclarePairedDelimiter\floor{\lfloor}{\rfloor}
  \usepackage{amsmath}
  \let\labelindent\relax
  \usepackage[shortlabels]{enumitem}
  \newcommand{\Mod}[1]{\ (\mathrm{mod}\ #1)}
  
  \newcommand{\argmax}{\arg\max}
  \newcommand\ar{\mathrm{A}}
  \newcommand\dr{\mathrm{D}}

  \newcommand\vo{\mathcal{V}}
  
  \newcommand\cd{c(\mathcal{G}^{\mathrm{D}}_k)}

  \newcommand\eam{\mathcal{E}^{\mathrm{A}}_m}
  \newcommand\edm{\mathcal{E}^{\mathrm{D}}_m}
  \newcommand\go{\mathcal{G}}
  \newcommand\ga{\mathcal{G}^{\mathrm{A}}_k}
  \newcommand\gd{\mathcal{G}^{\mathrm{D}}_k}
  \newcommand\eo{\mathcal{E}}
  \newcommand\ea{\mathcal{E}^{\mathrm{A}}_k}
  \newcommand\eob{\overline{\mathcal{E}}}
  \newcommand\eab{\overline{\mathcal{E}}^{\mathrm{A}}_k}
  \newcommand\eabm{\overline{\mathcal{E}}^{\mathrm{A}}_m}
  \newcommand\ed{\mathcal{E}^{\mathrm{D}}_k}

  \newcommand\ua{U^{\mathrm{A}}_l}
  \newcommand\ud{U^{\mathrm{D}}_l}

  \newcommand\ba{\beta^{\mathrm{A}}}
  \newcommand\bab{\overline{\beta}^{\mathrm{A}}}
  \newcommand\bd{\beta^{\mathrm{D}}}
  \newcommand\ra{\rho^{\mathrm{A}}}
  \newcommand\rd{\rho^{\mathrm{D}}}
  \newcommand\ka{\kappa^{\mathrm{A}}}
  \newcommand\kd{\kappa^{\mathrm{D}}}
  \newcommand\ha{h^{\mathrm{A}}}
  \newcommand\hd{h^{\mathrm{D}}}
  \newcommand\ta{T^{\mathrm{A}}}
  \newcommand\td{T^{\mathrm{D}}}
  \newcommand\lb{\alpha}
  \newcommand\lbb{\overline{l}^{\mathrm{D}}}
  \newcommand\la{l^{\mathrm{A}}}
  \newcommand\ld{l^{\mathrm{D}}}
  
\makeatletter
\usepackage{multirow}

\begin{document}
   \setlength{\abovedisplayskip}{4pt}
   \setlength{\belowdisplayskip}{6pt}
%
\title{\LARGE \bf Cluster Forming of Multiagent Systems \\ in Rolling Horizon Games with Non-uniform Horizons}
\author{Yurid Nugraha, Ahmet Cetinkaya, Tomohisa Hayakawa, Hideaki Ishii, and Quanyan Zhu
\thanks{Yurid Nugraha and Tomohisa Hayakawa are with the Department of  
Systems and Control Engineering, Tokyo Institute of Technology, Tokyo  
152-8552, Japan. {\tt\small{yurid@dsl.sc.e.titech.ac.jp,  
hayakawa@sc.e.titech.ac.jp}}}
\thanks{Ahmet Cetinkaya is with the Shibaura Institute of Technology, Tokyo, 135-8548, Japan. {\tt
\small{ahmet@shibaura-it.ac.jp}}}
\thanks{Hideaki Ishii is with the Department of Computer Science, Tokyo Insitute of Technology, Yokohama 226-8502,  Japan. {\tt
\small{ishii@c.titech.ac.jp}}}
\thanks{Quanyan Zhu is with the Department of Electrical and Computer Engineering, New York University, Brooklyn NY, 11201, USA. {\tt
\small{quanyan.zhu@nyu.edu}}}}
\maketitle

\begin{abstract}
Consensus and cluster forming of multiagent systems in the face of jamming attacks along with reactive recovery actions by a defender are discussed. The attacker is capable to disable some of the edges of the network with the objective to divide the agents into a smaller size of clusters while, in response, the defender recovers some of the edges by increasing the transmission power. We consider repeated games where the resulting optimal strategies for the two players are derived in a rolling horizon fashion. The attacker and the defender possess different computational
abilities to calculate their strategies. This aspect is represented by the non-uniform values of the horizon lengths and the game periods. Theoretical and simulation based results demonstrate the effects of the horizon lengths and the game periods on the agents' states.
\end{abstract}

           \begin{figure*}[t]
\centering
    \begin{minipage}[c]{0.60\textwidth}
    \centering
        \psfrag{ed1}{\scriptsize $\la=1$}
        \psfrag{ed2}{\scriptsize $\la=2$}
        \psfrag{ed3}{\scriptsize $\la=3$}
        \psfrag{ed4}{\scriptsize $\ld=1$}
        \psfrag{ed5}{\scriptsize $\ld=2$}
        \psfrag{ed6}{\scriptsize $\la=4$}
        \psfrag{ed7}{\scriptsize $\la=5$}
        \psfrag{ed8}{\scriptsize $\la=6$}
        \psfrag{ed9}{\scriptsize $\ld=3$}
        \psfrag{ed10}{\scriptsize $\ld=4$}
        \psfrag{ez9}{\scriptsize $\ha=6$}
        \psfrag{ez5}{\scriptsize $\hd=4$}
        \psfrag{ez4}{\scriptsize $\ta=2$}
        \psfrag{ez3}{\scriptsize $\td=3$}
        \psfrag{eda}{\footnotesize $1$st game}
        \psfrag{edb}{\footnotesize $2$nd game}
        \psfrag{edc}{\footnotesize $3$rd game}
        \psfrag{e4}{\small $13$}
        \psfrag{e3}{\small $12$}
        \psfrag{e1}{\small $10$}
        \psfrag{e2}{\small $11$}
        \psfrag{0}{\small $0$}
        \psfrag{1}{\small $1$}
        \psfrag{2}{\small $2$}
        \psfrag{3}{\small $3$}
        \psfrag{4}{\small $4$}
        \psfrag{5}{\small $5$}
        \psfrag{6}{\small $6$}
        \psfrag{7}{\small $7$}
        \psfrag{8}{\small $8$}
        \psfrag{9}{\small $9$}
        \psfrag{eza}{\footnotesize $1$st game}
        \psfrag{ezb}{\footnotesize $2$nd game}
    \includegraphics[trim=0 0 0 0.5cm,width=9 cm]{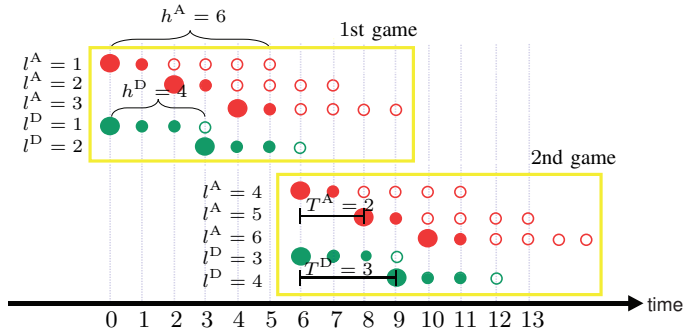}
        \vskip -5pt
    \end{minipage}
    \begin{minipage}[c]{0.25\textwidth}
        \vskip -30pt
        \caption{Sequence of games with decision-making indices $\la$ and $\ld$: attacker's horizon (red) and defender's horizon (green) with non-uniform game periods. The horizon lengths are $\ha=6$ and $\hd=4$, whereas the game periods are $T^{\mathrm{A}}=2$ and $T^{\mathrm{D}}=3$. There are two games denoted by the yellow rectangles; a game is played every $\mathrm{lcm}(\ta,\td)=6$ time instants in this example. The filled circles indicate the implemented strategies and the empty circles indicate the strategies of the game that are not implemented.}
        \label{nonT}
    \end{minipage}
    \vspace{-0.0cm}
\end{figure*} 

\begin{figure}[t]
\centering
    \begin{minipage}[c]{0.20\textwidth}
    \centering
    \includegraphics[trim=0 0 0 50,scale=0.3]{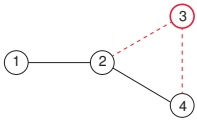}
        \vskip -5pt
    \end{minipage}
    \begin{minipage}[c]{0.25\textwidth}
        \vskip -10pt
        \caption{Example of node attacks discussed in Remark~\ref{rem1}. The attacker attacks node/agent 3, removing edges (3,2) and (3,4). As a result, agent~3 is disconnected from others.}
        \label{nod}
    \end{minipage}
    \vspace{-0.2cm}
\end{figure} 

\section{Introduction}
Multiagent systems are used to model the interaction between a number of agents capable of making local decisions in a network. Due to the distributed nature of the agents, they are prone to cyber attacks initiated by malicious adversaries \cite{sandb}. One of the most common form of cyber attacks is jamming attacks, where adversaries emit interference signals to prevent communication among agents. Jamming attacks on consensus problems of multiagent systems have been studied in, e.g., \cite{tesi,arxivYur}. 

In the presence of adversaries, agents following a standard consensus protocols may not always converge to the same state; instead, they may be divided into several clusters. Cluster forming in multiagent systems has been studied in, e.g., \cite{altafini}, where the weights in the agents' state updates may take negative values, representing possibly hostile relations among certain agents. Game theoretical approaches have been studied to analyze the interaction between such adversaries and the agents in networks \cite{li}.

\textcolor{black}{In an attack-prone multiagent system, model predictive control has been studied to address the situation where agents solve certain optimization problems with constraints by considering the future times characterized by the length of the horizons. A defense mechanism for agents working together under distributed model predictive control is studied in \cite{mpc1}. The rolling horizon concept, which is a key component of model predictive control, is considered in \cite{mzhu2} in a noncooperative security game related to replay attacks.}

In this paper, we consider a jamming attack and defense scenario in a two-player game setting between a centralized attacker and a centralized defender. The attacker attempts to divide the agents into as many clusters as possible, whereas the defender aims to keep the number of clusters small. This game is played repeatedly over time, where the players recalculate and may change their strategies according to a rolling horizon optimization approach. The players are assumed to have different computational abilities, represented by the non-uniform values of the \textit{horizon lengths} and the \textit{game periods}. This problem can be viewed as an extension of our previous studies \cite{cdc21,yurecc21}; we considered the special case with players' uniform horizon parameters in \cite{cdc21}, whereas the performance of the two players with non-uniform horizons was discussed in \cite{yurecc21}.

Games where players update their strategies in an asynchronous manner with different resources such as computation and data can be used to model real-life applications. For example, decisions involving firms from various countries cannot be done simultaneously due to different working times \cite{asyn,asyn2}. Due to the non-uniform horizons, the players' decision making process becomes complicated, executed at different time instants. The player with the longer horizon has a clear advantage; as we study the full information case,  this player may even solve the optimization problem of the opponent that will take place in the future. On the other hand, the player with the shorter horizon can no longer perfectly observe the opponent's planned action; this may result in a waste of the player's resources required to launch the attack/recovery actions. Similar energy allocation games in the context of cyber security have been discussed in, e.g., \cite{en1,en2}.

Here we consider energy allocation games in the context of cluster forming of agents. Our approach can be related to the concept of network effect/externality\cite{net}, where the utility of an agent in a certain cluster depends on how many other agents belong to that particular cluster. Such concepts have been used to analyze grouping of agents on, e.g., social networks and computer networks, as discussed in \cite{sns1,sns2}.

The paper is organized as follows. In Section~II, we introduce the framework for the rolling horizon game and energy consumption models of the players.
In Section~III, we describe in detail the structure of the game with non-uniform horizon lengths and game periods. We continue by discussing the theoretical results on consensus and cluster forming of agents in Sections~IV~and~V. We then provide numerical examples in Section~VI. Finally, we conclude the paper in Section~VII.

The notations used in this paper are fairly standard. We denote $|\cdot|$ as the cardinality of a set. The floor function and the ceiling function are denoted by $\floor{\cdot}$ and $\ceil{\cdot}$, respectively. The set of nonnegative integers is denoted by $\mathbb{N}_0$.

\section{Problem Formulation}
 We explore a multiagent system of $n$ agents communicating to each other in discrete time. The network topology is described by an undirected and connected graph $\go=(\mathcal{V},\mathcal{E})$. It consists of the set $\mathcal{V}$ of vertices representing the agents and the set $\mathcal{E} \subseteq \mathcal{V} \times \mathcal{V}$ of edges representing the communication links. Each agent $i$ has the scalar state $x_i$ following the consensus update rule at time $k \in \mathbb{N}_0$
    \begin{align} 
        x_i[k+1]&=x_i[k]+u_i[k], \label{state} \\
        u_i[k]&=\sum_{j \in {\mathcal{N}_{i}}[k]} a_{ij}(x_j[k]-x_i[k]), \label{state2}
    \end{align}
where $x[0]=x_0$, $a_{ij}>0$, $\sum_{j=1, j \neq i}^n a_{ij} < 1$, and $\mathcal{N}_i[k]$ denotes the set of agents that can communicate with agent $i$ at time $k$. This set may change due to the attacks. 

A two-player game between the attacker and the defender is considered. The attacker is capable to block the communication by jamming some targeted edges and therefore delay (or completely prevent) the consensus among agents. These jamming attacks are represented by the removal of edges in $\mathcal{G}$. In response, the defender tries to recover the inter-agent communications by allocating resources to rebuild some of those edges. 

We consider an attacker that has two types of jamming signals in terms of their strengths, \textit{strong} and \textit{normal}. The defender is able to recover only the edges that are attacked with normal strength; if the defender allocates its energy to the strongly-attacked edges, the edges cannot be rebuilt and the resources will be wasted. Similarly, if the defender allocates its resources to the edges that are not attacked, the resources will also be wasted without any improvement of the network connectivity. While the recent works \cite{li,arxivYur,cdc21} consider jamming in similar multiagent system settings, the notion of wasted resources does not appear there.

\subsection{Attack-recovery sequence}
    In our setting, the players make their attack/recovery actions at every time $k \in \mathbb{N}_0$. At the beginning of time $k$, the communication topology of the system is represented by $\go$. Then, the players decide to attack/recover certain edges in two stages, with the attacker acting first and then the defender. Hence, the game is that of a Stackelberg type.
    
    More specifically, at time $k$, the attacker attacks $\go$ by deleting $\ea\subseteq \eo$ with normal jamming signals and $\eab\subseteq \eo$ with strong jamming signals with $\ea \cap \eab = \emptyset$, whereas the defender recovers $\ed \subseteq \mathcal{E}$. Due to the attacks and then the recoveries, the network changes from $\go$ to $\ga:=(\mathcal{V},\eo \setminus (\ea\cup \overline{\mathcal{E}}^{\mathrm{A}}_k))$ and further to $\gd:=(\mathcal{V},\eo \setminus (\ea\cup \overline{\mathcal{E}}^{\mathrm{A}}_k) \cup (\ed \cap \ea))$. The agents then communicate to their neighbors based on this resulting graph $\gd$. 
    
    In this game, the players attempt to choose the best strategies in terms of edges attacked/recovered $(\overline{\mathcal{E}}^{\mathrm{A}}_k,\ea)$ and $\ed$ to maximize their own utility functions. Here the game is defined over the horizon of several steps. The players make decisions in a rolling horizon fashion as explained more in Section \ref{fiv}; the optimal strategies that have been obtained at a past time may change when the players recalculate their strategies at a future time. Fig.~\ref{nonT} illustrates the discussed sequence over time; the attacker's and the defender's \textit{horizon lengths}, i.e., how far in the future the players look ahead when determining their strategies, are denoted by $\ha$ and $\hd$, respectively, whereas the \textit{game periods}, i.e., how often players update their strategies, are denoted by $\ta$ and $\td$ (discussed in more detail later). As a consequence of having non-uniform game periods, players have separate decision-making processes represented by the decision-making indices $\la$ and $\ld$; a game is defined as a set of decision-making process that starts from a time where the players simultaneously update their strategies, indicated by the yellow boxes.
    
    {\color{black}\begin{remark}\label{rem1}
    In addition to the attacks and the recoveries based on individual edges as introduced above, we can consider a slightly different setting where the attacker can attack nodes/agents so that \textit{all} edges adjacent to the attacked agents are disconnected, as shown in Fig.~\ref{nod}. Specifically, the attacker's actions are now $\ea\in\mathcal{F}$ and $\eab\in\mathcal{F}$, where $\mathcal{F}:=\{\emptyset, F_1,F_2,\ldots,F_n,F_{1}\cup F_{2},\ldots,\eo=\cup_{i\in\vo}F_i\}$ and $F_i:=\{(i,j): \ j\in \mathcal{N}_i\}$ represents the set of edges adjacent to agent~$i$. In this case, the attacker effectively attacks a node/agent by attacking all edges adjacent to it. 
    \end{remark}}
    
    \vspace{-0.1cm}
    \subsection{Energy constraints}
    By attacking and recovering, the players allocate their energies to the attacked/recovered edges. These actions are affected by the constraints on the energy resources, which increase linearly in time. The energy consumed by the players is proportional to the number of planned attacked/recovered edges as well. Here, the strong attacks on $\eab$ take $\bab>0$ energy per edge per unit time compared to the normal attacks on $\ea$, which take $\ba>0$ energy per edge, where $\bab>\ba$. The total energy used by the attacker by time $k$ is constrained as
    \begin{align}
    \sum_{m=0}^{k} & (\bab|\eabm|+\ba|\eam| ) \leq \ka + \ra k, \label{a}
    \end{align} 
     where $\kappa^{\mathrm{A}}\geq \ra>0$, $\ba>0$. This inequality implies that the total energy spent by the attacker cannot exceed the available energy characterized by the initial energy $\kappa^{\mathrm{A}}$ and the supply rate $\rho^{\mathrm{A}}$. This energy constraint upper-bounds the number of edges that the attacker can attack.

    The energy constraint of the defender is similar to (\ref{a}) and is given by
    \begin{equation} \label{en.d}
        \sum_{m=0}^{k} \bd|\edm| \leq \kd + \rd k
    \end{equation}
    with $\kd\geq \rd>0$, $\bd>0$. Note that the defender may allocate its energy inefficiently, i.e., the defender may attempt to recover unattacked edges or edges attacked with strong jamming signals.
    
    {\color{black}{For the node attack case discussed in Remark~\ref{rem1}, the energy constraint of the attacker (\ref{a}) becomes $\sum_{m=0}^{k} (\overline{\beta}^{\ar}_{\vo}|\overline{\vo}^{\ar}_m|+\beta^{\ar}_{\vo}|{\vo}^{\ar}_m|) \leq \ka + \ra k$} with energies $\overline{\beta}^{\ar}_{\vo}>{\beta}^{\ar}_{\vo}$, where $\overline{\vo}^{\ar}_m$ and ${\vo}^{\ar}_m$ denote the sets of nodes/agents whose adjacent edges are attacked with strong and normal jamming signals, respectively. Note that in this case, if an edge is attacked by \textit{both} normal signals and strong signals, then that edge cannot be recovered by the defender.} 
    
    \subsection{Agent clustering and state difference}
    By attacking, the attacker makes the graph disconnected and separates the agents into clusters (i.e., sets of agents). We introduce a few notions related to grouping/clustering of agents. We call each subset $\mathcal{C}\subseteq\mathcal{V}$ of agents taking the same state at infinite time as a \textit{cluster}, i.e., $\lim_{k \to \infty} x_{i}[k]=\lim_{k \to \infty} x_{j}[k], \quad \forall i,j\in \mathcal{C}$.
    
    In the considered game, the attacker and the defender are concerned about the number of agents in each group. Specifically, we follow the notion of \textit{network effect/network externality}\cite{net}, where the utility of an agent in a certain cluster depends on how many other agents belong to that particular cluster. In the context of this game, the attacker attempts to isolate agents so that fewer agents are in each group, while the defender wants as many agents as possible in the same group. We then represent the level of grouping in the graph $\mathcal{G}'$ by the function $c(\cdot)$, which we call the \textit{agent-group index}, given by 
    \begin{align}\label{cluster}
    c(\mathcal{G}'):=\sum_{p=1}^{\overline{n}(\mathcal{G}')} |\vo'_p|^2 -|\vo|^2 \quad (\leq 0).
    \end{align}
    The value of $c(\mathcal{G}')$ is 0 if $\mathcal{G}'$ is connected, since there is only one group (i.e., $\overline{n}(\mathcal{G}')=1$). A larger value (closer to 0) of $c(\mathcal{G}')$ implies that there are fewer groups in graph $\mathcal{G}'$, with each group having more agents.
    
    In our problem setting, the players also consider the effects of their actions on the agent states when attacking/recovering. For example, the attacker may want to separate agents having state values with more difference in different groups. We specify the sum of the agents' state differences $z_k$ of time $k$ as
    \begin{align}
    z_k(\eab,\ea,\ed):= x^{\mathrm{T}}[k+1] L_{\mathrm{c}} x[k+1], \label{z}
    \end{align}
    with $L_{\mathrm{c}}$ being the Laplacian matrix of the complete graph with $n$ agents. The attacked and recovered edges $(\eab,\ea,\ed)$ will affect $x[k+1]$, and in turn influence the value of $z_k$. Note that the value of $z_k$ does not increase over time \cite{FB-LNS} because of the protocol given in (\ref{state}) and (\ref{state2}) even if the system is under attacks.

    The game structure, explained in more detail later, is illustrated in Fig.~\ref{nonT}. The attacker's (resp., the defender's) utility functions of the $\la$th (resp., $\ld$th) decision-making index with $\la,\ld \in \mathbb{N}$ starting at time $k=(\la-1)\ta$ (resp., $k=(\ld-1)\td$) take account of the agent-group index $c(\cdot)$ over time horizons $\ha,\hd\geq1$ from time $(\la-1)\ta$ to $(\la-1)\ta+\ha-1$ (resp., from $(\ld-1)\td$ to $(\ld-1)\td+\hd-1$). Specifically, the utility functions at the $\la$th decision-making process for the attacker and at the $\ld$th decision-making process for the defender are
   \begin{align}
    U^{\ar}_{\la}  &:= \sum_{k=(\la-1)\ta}^{(\la-1)\ta+\ha-1} (a z_k-b \cd), \label{ua} \\
    U^{\dr}_{\ld} &:= \sum_{k=(\ld-1)\td}^{(\ld-1)\td+\hd-1} (-a z_k+ b \cd) \label{ud},
    \end{align}
    which are to be maximized by the players. The player with a longer horizon length and a shorter game period is expected to use its energy more efficiently, and thus to obtain a higher utility over time.

    \vspace{-0.1cm}
    \section{Game Structure with Non-uniform Rolling Horizon Lengths and Game Periods} \label{fiv}
    We are interested in finding the subgame perfect equilibrium of the game outlined so far. To this end, the game is divided into some subgames/decision-making points. The subgame perfect equilibrium must be an equilibrium in every subgame. The optimal strategy of each player is obtained by using a backward induction approach, i.e., by finding the equilibrium from the smallest subgames. The tie-break condition happens when the players' strategies result in the same utility.  In this case, we suppose that the players choose to attack/recover more edges if they have enough energy to attack/recover all edges at all subsequent steps; otherwise, they will attack/recover fewer edges.
    
    In this section, before considering the more general setting, we consider a simpler scenario. This case is when the players employ different horizon parameters but their game periods are the same. Then, we study the case when the game periods of the players also differ. The first case still represents players with different computational abilities to solve games.

    \subsection{Non-uniform Horizon Lengths}
    
        \begin{figure*}
 \begin{minipage}[c]{0.70\textwidth}
        \centering
        \psfrag{A1}{\scriptsize $\overline{\eo}^{\ar}_{l,1},\eo^{\ar}_{l,1}$}
        \psfrag{A2}{\scriptsize $\overline{\eo}^{\ar}_{l,2},\eo^{\ar}_{l,2}$}
        \psfrag{A3}{\scriptsize $\eo^{\dr}_{l,2}$}
        \psfrag{A4}{\scriptsize $\eo^{\dr}_{l,1}$}
        \psfrag{A5}{\scriptsize $\overline{\eo}^{\ar}_{l,3},\eo^{\ar}_{l,3}$}
        \psfrag{A6}{\scriptsize $\eo^{\dr}_{l,3}$}
        \includegraphics[scale=0.55]{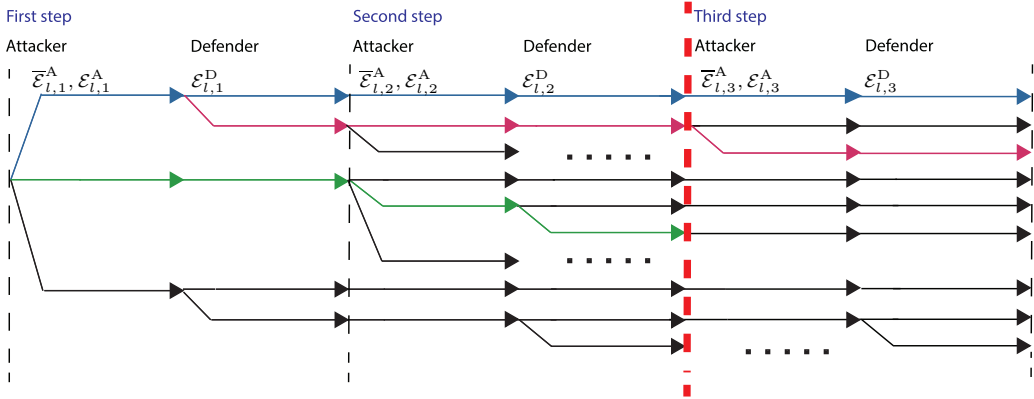}
        \vspace{-0.5cm}
    \end{minipage}
    \hfill
    \begin{minipage}[c]{0.20\textwidth}
        \caption{Extensive-form game for different $\ha$ and $\hd$. The vertical dashed lines denote the different steps of the game, whereas the dashed red line denotes boundary of different player's horizon length. The optimization beyond this limit is done by only the player with longer horizon (in this case the attacker).}
        \label{difhor}
    \end{minipage}
    \vspace{-0.5cm}
    \end{figure*}
    
    In this subsection, we explain the game structure with non-uniform horizon length and uniform game periods. That is, we assume $\ta = \td = T$. This implies that both players make their decisions periodically at the same time. Thus, the indices for the decision-making processes of the players  become equal $\la=\ld=l$ at all times.

    Due to the nature of the rolling horizon approach, the strategies obtained for the $l$th decision-making process, i.e., attacked and recovered edges, are applied only from time $(l-1)T$ to $lT-1$ with $T\leq \min\{\ha,\hd\}$. The players' strategies at the $l$th decision-making process are specified as $((\eob^{\ar}_{l,1},{\eo}^{\ar}_{l,1},\eo^{\dr}_{l,1}),\ldots,(\eob^{\ar}_{l,\hd},{\eo}^{\ar}_{l,\hd},{\eo}^{\dr}_{l,\hd}),(\eob^{\ar}_{l,\hd+1},{\eo}^{\ar}_{l,\hd+1}),$ $\ldots,(\eob^{\ar}_{l,\ha},{\eo}^{\ar}_{l,\ha}))$ if $\ha>\hd$, and $((\eob^{\ar}_{l,1},\eo^{\ar}_{l,1},\eo^{\dr}_{l,1}),$ $\ldots,(\eob^{\ar}_{l,\ha},\eo^{\ar}_{l,\ha},\eo^{\dr}_{l,\ha}),\eo^{\dr}_{l,\ha+1},\ldots,\eo^{\dr}_{l,\hd})$ if $\ha<\hd$, with $\overline{\mathcal{E}}^{\mathrm{A}}_{l,\alpha},\eo^{\mathrm{A}}_{l,\alpha},\eo^{\mathrm{D}}_{l,{\alpha}}$ indicating the strategies at the $\alpha$th step of the $l$th decision-making process with $\alpha \in \mathbb{N}$. Note that if $\ha>\hd$, only the attacker formulates its strategies after $\hd$th step. Similarly, if $\ha<\hd$, only the defender formulates its strategies after $\ha$th step. The case where $\ha=\hd$ can be similarly handled, with the strategies $((\eob^{\ar}_{l,1},{\eo}^{\ar}_{l,1},\eo^{\dr}_{l,1}),\ldots,(\eob^{\ar}_{l,\ha},{\eo}^{\ar}_{l,\ha},{\eo}^{\dr}_{l,\ha}))$.
    
    From these obtained strategies, only the strategies from the 1st step to the $T$th step are applied. Since we consider the full information setting, the values of $\ha$ and $\hd$ are known to both players.

    We now provide an example to explain how the optimal edges are obtained for the case of $\ha=3$ and $\hd=2$. The optimal strategies of the players solved backward in time at the decision-making index $l$ are given by:
\begin{itemize}
     \vspace{-0.2cm}
    \item Step 3:
    \vspace{-0.2cm}
      \begin{align}
      &\eo^{\dr*}_{l,3} (\eob^{\ar}_{l,3},\eo^{\ar}_{l,3})  \in \argmax_{\eo^{\dr}_{l,3}}  -U^{\ar}_{l,3}, \\
    &(\eob^{\ar*}_{l,3}(\eo^{\dr}_{l,2}),\eo^{\ar*}_{l,3}(\eo^{\dr}_{l,2})) \in \argmax_{(\eob^{\ar}_{l,3},\eo^{\ar}_{l,3})} U^{\ar}_{l,3} (\eo^{\dr*}_{l,3}),  \notag \end{align}
    \vspace{-0.7cm}
    \begin{align}
    \textcolor{white}{-}
    \end{align}
    \vspace{-0.2cm}
    \item Step 2:
    \vspace{-0.2cm}
    \begin{align}
    &\eo^{\dr\prime}_{l,2}(\eob^{\ar}_{l,2},\eo^{\ar}_{l,2}) \in \argmax_{\eo^{\dr}_{l,2}}  U^{\dr}_{l,2},\\
    &(\eob^{\ar{\prime}}_{l,2}(\eo^{\dr}_{l,1}),\eo^{\ar{\prime}}_{l,2}(\eo^{\dr}_{l,1})) \in \argmax_{(\eob^{\ar}_{l,2},\eo^{\ar}_{l,2})} -U^{\dr}_{l,2}(\eo^{\dr\prime}_{l,2}), \notag \\[-10pt] \label{in2}
    \end{align}
    \begin{align}
    & (\eob^{\ar*}_{l,2}(\eo^{\dr}_{l,1}),\eo^{\ar*}_{l,2}(\eo^{\dr}_{l,1})) \in  \argmax_{(\eob^{\ar}_{l,2},\eo^{\ar}_{l,2})} U^{\ar}_{l,2}(\eo^{\dr\prime}_{l,2}), \notag \end{align}
    \vspace{-0.9cm}
    \begin{align}
    \textcolor{white}{-} \label{in3}
    \end{align} 
    \vspace{-0.2cm}
    \item Step 1:
    \vspace{-0.2cm}
    \begin{align}
    &\eo^{\dr\prime}_{l,1}(\eob^{\ar}_{l,1},\eo^{\ar}_{l,1}) \in \argmax_{\eo^{\dr}_{l,1}} \ud(\eob^{\ar'}_{l,2},\eo^{\ar\prime}_{l,2}),
    \label{in1} \\
    & (\eob^{\ar*}_{l,1},\eo^{\ar*}_{l,1}) \in \argmax_{(\eob^{\ar}_{l,1},\eo^{\ar}_{l,1})} \ua(\eo^{\dr\prime}_{l,1}), \label{in0}
    \end{align}
\end{itemize}
    where $U^{\ar}_{l,\alpha}:=\sum_{k=(l-1)T+\alpha-1}^{(l-1)T+ \ha-1} = a z_k - b\cd $ (resp., $U^{\dr}_{l,\alpha}:=\sum_{k=(l-1)T+\alpha-1}^{(l-1)T+\hd-1} = - a z_k + b\cd $) is defined as parts of $\ua$ (resp., $\ud$) calculated from the $\alpha$th step to the $\ha$th (resp., $\hd$th) step of the $l$th decision-making process.
    
    These optimization problems are solved backward every game period $T$ from the $(\max\{\ha,\hd\})$th step of the $l$th decision-making process. Note that to find $(\eob^{\ar*}_{l,1},\eo^{\ar*}_{l,1})$, one needs to obtain $(\eo^{\dr*}_{l,1}(\eob^{\ar}_{l,1},\eo^{\ar}_{l,1}))$ beforehand. Likewise, to find $(\eo^{\dr*}_{l,1}(\eob^{\ar}_{l,1},\eo^{\ar}_{l,1}))$, one needs to obtain $(\eob^{\ar*}_{l,2}(\eo^{\dr}_{l,1}),\eo^{\ar*}_{l,2}(\eo^{\dr}_{l,1}))$, and so on. Also, note that while $\eo^{\dr*}_{l,3}$ is not part of the defender's strategy, it is still needed for the attacker to obtain $(\eob^{\ar*}_{l,3},\eo^{\ar*}_{l,3}$). Therefore, outside the defender's ability characterized by its horizon length $\hd$, here we suppose that the attacker utilizes the strategy that emulates the defender's best response with longer horizon, i.e., from part of the utility function $-\ua$. Throughout this paper, we denote $(\eob^{\ar*}_{l,\lb},\eo^{\ar*}_{l,\lb},\eo^{\dr*}_{l,\lb})$ as the optimal strategies according to the player with longer horizon length and $(\eob^{\ar\prime}_{l,\lb},\eo^{\ar\prime}_{l,\lb},\eo^{\dr\prime}_{l,\lb})$ as the optimal strategies according to the player with shorter horizon length.
    
    In the step $\alpha\leq\hd (< \ha)$, the defender assumes that the attacker's optimal edges, e.g., in (\ref{in2}), are based on the defender's utility function, which consists of $\hd$ steps only. The defender's optimal strategies according to the attacker, e.g., in (\ref{in1}), are based on the defender's perception of the attacker's optimal strategies, i.e., $(\eob^{\ar\prime}_{l,2},\eo^{\ar\prime}_{l,2})$, since the defender is not able to foresee the attacker's strategy beyond $\hd$. For the attacker, since it is able to compute the optimal strategy for the defender as well (due to the longer $\ha$), the attacker's strategies in the steps with index $\alpha\leq\hd$, e.g., (\ref{in3}) and (\ref{in0}), are based on $\eo^{\dr \prime}_{l,\alpha}$.
    
    In this Stackelberg game setting, the defender's strategy space depends on the attacker's strategy at the same step, i.e., the defender can only recover edges attacked normally. Hence, it is possible that the defender cannot perfectly apply its strategy. Specifically the defender may not recover some of $\ed$, in which the energy is allocated, when the attacker changes its own strategy. In this case, the defender will apply the strategy only on the edges that can be recovered. However, as explained above, it is natural that the recovery of the edges not attacked still consumes energy. This will be important to the discussion of consensus and clustering as shown later. 
    
    The decision-making process of the players in this example is illustrated in the game tree in Fig.~\ref{difhor}, where the blue line indicates the \textit{equilibrium path}, i.e., the strategy taken by the players following backward induction, for $\ha=\hd=3$. The green line indicates the equilibrium path for $\ha=\hd=2$, and the magenta line indicates the equilibrium path for $\ha=3,\hd=2$. In step 2, the attacker assumes that $\mathcal{E}^{\mathrm{D}*}_{l,2}$ comes from the utility over $h^{\mathrm{A}}=3$. The case where $\ha<\hd$ can be similarly described.

\subsection{Non-uniform Game Periods}
\vspace{-0.1cm}
In this subsection, we extend our discussion to the case of non-uniform game periods $\ta$ and $\td$ for the attacker and the defender, respectively. The corresponding decision-making indices are $\la$ and $\ld$, which respectively consist of $\lb^{\ar}$ and $\lb^{\dr}$ steps. These periods $\ta$ and $\td$ are known by both players as we consider the full information case. To ensure that both players are able to obtain their own strategies at any $k$, we set $\ta\leq \ha$ and $\td \leq \hd$. 

The game with non-uniform game periods is illustrated in Fig.~\ref{nonT}. The yellow rectangle indicates the set of \textit{decision-making processes} in one game, which follows a certain pattern. A game is played, i.e., both players simultaneously update their strategies, every lowest common multiple of $\ta$ and $\td$ denoted as $\mathrm{lcm}(\ta,\td)$; in Fig.~\ref{nonT}, the game is played every 2 time steps. With this formulation, it is expected that the players have better performance with shorter game period.

From Fig.~\ref{nonT}, we see that the players may not decide their strategies at the same time. For example, at time $k=2$, only the attacker updates its strategy, whereas the defender does not due to longer $\td$. Since $\ta$ and $\td$ are known by both players, at $k=2$ the attacker decides its strategy considering the defender's strategy that is obtained before at $k=0$. Furthermore, since $\ha=6$ in Fig.~\ref{nonT}, here the attacker with the ability to compute for three time steps ahead can only foresee three steps forward for the defender's $2$nd decision-making process.

Since the non-uniform game periods make the players decide their strategies at different times, we use different decision-making indices $\la$ and $\ld$ to specify the decision-making processes that occur at times $(\la-1)\ta$ and $(\ld-1)\td$ for the attacker and the defender, respectively, where the players maximize the utility functions (\ref{ua}) and (\ref{ud}). Note that different values of these indices for the players may refer to the same time step; e.g., in Fig.~\ref{nonT}, both $\la=2$, $\lb^{\ar}=1$ and $\ld=1$, $\lb^{\dr}=3$ correspond to $k=2$. 

As the players decide their strategies at different times, the optimization problems are different in each time. For example, the optimal strategy of the attacker at time $k=2$ in the case shown in Fig.~\ref{nonT} is given by (only some steps are shown due to space limitation):
\begin{itemize}
    \vspace{-0.1cm}
    \item Step 6 ($k=7$, Step 1 for defender):
    \vspace{-0.2cm}
    \begin{align}
    &\eo^{\dr*}_{3,2} (\eob^{\ar}_{2,6},\eo^{\ar}_{2,6})  \in \argmax_{\eo^{\dr}_{3,2}}  -U^{\ar}_{2,6}, \label{in14}\\
    &(\eob^{\ar*}_{2,6}(\eo^{\dr}_{2,4}),\eo^{\ar*}_{2,6}(\eo^{\dr}_{2,4})) \in \argmax_{(\eob^{\ar}_{2,6},\eo^{\ar}_{2,6})} U^{\ar}_{2,6} (\eo^{\dr*}_{3,2}), \notag
    \end{align}
    \vspace{-0.7cm}
    \begin{align}
    \textcolor{white}{-}
    \end{align}
    \vspace{-0.4cm}
    \item Step 5 ($k=6$, Step 4 for defender): \vspace{-0.2cm}
    \begin{align}
    &\eo^{\dr\prime}_{2,4}(\eob^{\ar}_{2,5},\eo^{\ar}_{2,5}) \in \argmax_{\eo^{\dr}_{2,4}}  U^{\dr}_{2,4}, \label{in16} \\
    & (\eob^{\ar*}_{2,5}(\eo^{\dr}_{2,3}),\eo^{\ar*}_{2,5}(\eo^{\dr}_{2,3})) \in  \argmax_{(\eob^{\ar}_{2,5},\eo^{\ar}_{2,5})} U^{\ar}_{2,5}(\eo^{\dr\prime}_{2,4}), \notag 
    \end{align}\vspace{-0.7cm}
    \begin{align}
    \label{in23}
    \end{align}\vspace{-1cm}
    \begin{align}
    \vspace{-0.2cm}
    \vdots \notag
    \end{align}
    \item Step 2 ($k=3$, Step 1 for defender): \vspace{-0.1cm}
    \begin{align}
    &\eo^{\dr\prime}_{2,1}(\eob^{\ar}_{2,2},\eo^{\ar}_{2,2}) \in \argmax_{\eo^{\dr}_{2,1}}  U^{\dr}_{2,1} (\eob^{\ar\prime}_{2,3},\eo^{\ar\prime}_{2,3}), \label{in16b} \\
    & (\eob^{\ar*}_{2,2}(\eo^{\dr}_{2}),\eo^{\ar*}_{2,2}(\eo^{\dr}_{2})) \in  \argmax_{(\eob^{\ar}_{2,2},\eo^{\ar}_{2,2})} U^{\ar}_{2,2}(\eo^{\dr\prime}_{2,1}), \notag 
    \end{align}
    \vspace{-0.7cm}
    \begin{align}
    \textcolor{white}{-} \label{in23b}
    \end{align}
    \vspace{-0.4cm}
    \item Step 1 $(k=2)$: \vspace{-0.2cm}
    \begin{align}
    & (\eob^{\ar*}_{2,1},\eo^{\ar*}_{2,1}) \in \argmax_{(\eob^{\ar}_{2,1},\eo^{\ar}_{2,1})} U^{\ar}_2(\eo^{\dr}_2). \label{in20}
    \end{align}
\end{itemize}

    The attacker cannot compute more than $\ha=6$ time steps ahead, and hence in (\ref{in14}) above the attacker will use its own utility function $U^{\ar}_{\la}$ to estimate the defender's optimal edges at $\ld=3$. By $k \Mod \td \ne 0$, the defender does not make a new decision and thus will apply the strategy obtained in the previous time instead, e.g., $\eo^{\dr}_2$ obtained at $k=0$. Therefore, it is then possible for the player with shorter game period (in this case, the attacker) to benefit by changing its strategies; for example, in the case explained above, the attacker may benefit by changing $\eo^{\ar}_2$ to avoid the recovery by the defender in $\eo^{\dr}_2$, which has been set and cannot be changed. In this game, it is assumed that the player with longer horizon length is able to correctly recall the strategy of the opponent that has been determined at the same time as its own strategy, i.e., at $i \mathrm{lcm}(\ta,\td), i\in\mathbb{N}_0$. For example, the attacker in this case knows $\eo^{\dr}_2$ since it is decided at time $k=0$, i.e., the same time as the attacker's decision-making time. The strategies of the shorter horizon player determined in the past but not at the same time are also known by the longer horizon player as long as the entire horizon of the shorter horizon player falls into the same horizon of the longer horizon player.
    
    

    \vspace{-0.1cm}
    \section{Consensus Analysis}\label{thre}
    {\color{black}We now examine the effect of the game structure and players' energy constraints on consensus. We note that our earlier work \cite{yurecc21} dealt with players' performance given non-uniform horizons, but consensus and cluster forming were not discussed there, and utilities had different forms.
    
    We first investigate the defender's optimal strategies. 
    \vspace{0.1cm}
    
    \begin{lemma} \label{lem01}
    There exists an infinite sequence  $\lbb:=\{\lbb_1,\lbb_2,\ldots\}$ of the defender's decision-making indices where $\lbb_{i+1}>\lbb_{i}$ and $\lbb_i\in \mathbb{N}$ such that in the $\lbb_i$th decision-making process, the optimal strategy for the defender in the first step is to recover $\eo^{\dr}_{\lbb_i,1}\neq \emptyset$ as long as $\eo^{\ar}_{\lbb_i,1}\neq \emptyset$.
    \end{lemma}
    \begin{proof}
    We note that if the defender does not recover from nonzero normal attacks in the first step of the game, the worst scenario is that agent states will eventually converge to different values. Consequently, the attacker needs to keep attacking the edges connecting the agents with different states to keep them separated from other clusters. Suppose that the agents are separated into clusters at the game with index $\lbb_i$. Here, we can verify that it always holds:
    \begin{align}
    z_{\lbb_i,1}(\cdot,\tilde{\eo}^{\ar}_{\lbb_i,1},\emptyset)\geq z_{\lbb_i,1}(\cdot,\tilde{\eo}^{\ar}_{\lbb_i,1},\eo^{\dr}_{\lbb_i,1}), \label{sy}    
    \end{align}
    with $\tilde{\eo}^{\ar}_{\lbb_i}$ being the edges separating agents with different states. In this case, the attacker needs to attack $\tilde{\eo}^{\ar}_{\lbb_i}$ to keep the agents from arriving at consensus. Note that (\ref{sy}) is a more specific form of $ z_k(\emptyset,\eo,\emptyset)\geq z_k(\emptyset,\eo,\eo^{\dr}_k)$, where attacking all edges always gives the maximum value of $z_k$.
    
    From (\ref{sy}), the defender always benefits from recovering nonzero number of edges, since $z_{\lbb_i,2}=z_{\lbb_i,1}$ if the defender does not recover any edge, which gives the least value of utility. Since $c(\go^{\dr}_{\lbb_i,1})$ also gives the lowest value if the defender does not recover, at the first step of the ${\lbb_i}$th game the defender's utility with recovering nonzero edges is always better than the case of not recovering any edge. Since the defender constantly gains $\rd$ amount of energy at each time, this action for the defender is the same for the next games with indices $\lbb_{i+1}$, $\lbb_{i+2}$, and so on.
    \end{proof}
    \vspace{0.1cm}
    
        \vspace{0.1cm}
        
    The following two propositions provide necessary conditions for the agents to be separated into multiple clusters for infinitely long duration without achieving consensus. 
    
    \begin{proposition} \label{lem3.4}
      A necessary condition for consensus not to occur is $\ra/\ba \geq \lambda$, where $\lambda$ is the connectivity of $\mathcal{G}$.
    \end{proposition}
    \begin{proof}
    We note that, without any recovery from the defender $(\eo^{\mathrm{D}}_k=\emptyset)$, the attacker must attack at least $\lambda$ number of edges with normal signals at any time $k$ in order to make $\mathcal{G}^{\mathrm{D}}_k$ disconnected. If the attacker attacks $\lambda$ edges with normal jamming signals at all times, the energy constraint (\ref{a}) becomes $(\ba \lambda-\ra) k \leq \ka$. Thus, the condition $\ra/\ba\geq \lambda$ has to be satisfied for all $k$.
    \end{proof}
    
  \begin{proposition} \label{lem3}
    A necessary condition for consensus not to occur is $\ra/\bab \geq \lambda$ if either of the following two conditions is satisfied:
    \begin{enumerate}[(a)]
            \item $b=0$, $\hd\geq\ha$ and $\mathrm{lcm}(\ta,\td)=\ta$; or
            \item $b=0$ and $\td=1$.
    \end{enumerate}
    \end{proposition}
    \begin{proof}
    We prove by contrapositive; especially, we prove that consensus always happens if $\ra/\bab<\lambda$ under the specified conditions.
    
    We first suppose that the attacker attempts to attack $\lambda$ edges strongly at all times to disconnect the graph $\gd$. From (\ref{a}), the energy constraint of the attacker at time $k$ becomes
    $(\bab \lambda-\ra) k \leq \ka$.
    This inequality is not satisfied for higher $k$ if $\ra/\bab < \lambda$, since the left-hand side becomes positive and $\ka$ is finite. Therefore, the attacker cannot attack $\lambda$ edges strongly at all times if $\ra/\bab < \lambda$, and is forced to disconnect the graph by attacking with normal jamming signals instead. As a consequence, it follows from Lemma~\ref{lem01} that there exists an interval of time where the defender always recovers, i.e., $\eo^{\dr}_{\lbb_i,1}\neq \emptyset$, $i=1,2,\ldots,$ are optimal given that $\eo^{\ar}_{{\lbb_i,1}} \ne \emptyset$. Note that this strategy is always applied since it is for the first step of the game.

     From the utility function in (\ref{ud}), given that $b=0$, we can see that the defender obtains a higher utility if the agents are closer, which means that given a nonzero number of edges to recover (at the first step of the games with index $\lbb_i$ described above), the defender recovers the edges connecting further agents. Specifically, for the sequence of decision-making indices $[\lbb_i,\lbb_{i+1}]$, there is a time step where $U^{\dr}_{\ld}(\eo^{\dr}_{\ld,1}=\eo_1) \geq U^{\dr}_{\ld}(\eo_2)$, where $\eo_1$ and $\eo_2$ denote the sets of edges connecting agents with further states and closer states, respectively. Since by communicating with the consensus protocol (\ref{state}) the agents' states are getting closer, the defender will choose different edges to recover if the states of the agents connected by the recovered edges $\ed$ become close enough.
     
     For Case~(a), with $\hd\geq\ha$ and $\mathrm{lcm}(\ta,\td)=\ta$, it is guaranteed that the defender does not waste any energy by recovering, since by having longer horizon length the defender will accurately predict the attacker's action. The game period $\mathrm{lcm}(\ta,\td)=\ta$ implies that the attacker will never update its decision alone, i.e., the defender also updates when the attacker updates, which prevents the attacker to unilaterally change its strategy to avoid the defender's planned recovery. On the other hand, for Case~(b), the defender with $\td=1$ will be able to perfectly observe the attacker's action, and hence can fully avoid the possibility of wasting energy. 
     
     Consequently, for both Cases (a) and (b), if $\ra/\bab < \lambda$, then there exists $j \in \mathbb{N}$ depending on $i$ such that the union of graphs, i.e., the graph having the union of the edges of each graph $(\vo,\bigcup((\eo\setminus(\eab \cup \ea))\cup\ed))$, over the decision-making index $[\lbb_i,\lbb_{i+j(i)}]$ becomes a connected graph for all~$i$. These intervals $[(\lbb_i-1)\td,(\lbb_{i+j(i)}-1)\td]$, $i=1,2,\ldots,$ occur infinitely many times, since the defender's energy bound keeps increasing over time.

    It is shown in \cite{ren} that with protocol (\ref{state}), the agents achieve consensus in the time-varying graph if the union of the graphs over bounded time intervals is a connected graph. This implies that consensus is achieved if $(\vo,\bigcup((\eo\setminus(\eab \cup \ea))\cup\ed))$ is connected over $[\lbb_i,\lbb_{i+j(i)}]$ for all~$i$.
    \end{proof}
    \vspace{0.1cm}

    The next result provides a condition for consensus to be completely blocked and all agents are separated from each other. It shows that the attacker should be capable to make strong attacks on all the edges for all time. 
    
    \begin{proposition} \label{lem3.5}
   A sufficient condition for all agents not to achieve consensus at infinite time  is $\ra/\bab \geq |\eo|$.
    \end{proposition}
    \begin{proof}
    As it holds $z_{\la,\lb^{\ar}}(\mathcal{E},\emptyset,\emptyset)\geq z_{\la,\lb^{\ar}}(\eob^{\ar}_{\la,\lb^{\ar}},$ $\eo^{\ar}_{\la,\lb^{\ar}},\eo^{\dr}_{\la,\lb^{\ar}})$ and $c((\vo,\emptyset))\geq c((\vo,\eo \setminus (\eob^{\ar}_{\la,\lb^{\ar}}\cup \eo^{\ar}_{\la,\lb^{\ar}}) \cup (\eo^{\dr}_{\la,\lb^{\ar}}\cap \eo^{\ar}_{\la,\lb^{\ar}})))$, the function $U^{\mathrm{A}}_{\la}$ has the highest value if the attacker attacks all edges, i.e., $\eob^{\ar}_{\la,\lb^{\ar}} = \mathcal{E}$ or $\eo^{\ar}_{\la,\lb^{\ar}} = \mathcal{E}$. With $\ra/\bab \geq |\mathcal{E}|$, the attacker can attack all edges of $\mathcal{G}$ with strong jamming signals at any time $k$. Thus, the attacker will attack $\eo$ strongly at the $\lb^{\ar}$th step of the $\la$th decision-making process, i.e., $\eob^{\ar*}_{\la,\lb^{\ar}}=\mathcal{E}$, which also prevents the defender from recovering any edge, i.e., $\eo^{\dr}_{\la,\lb^{\ar}}=\emptyset$, for all $\la,\lb^{\ar}$. Thus the attacker will attack $\mathcal{E}$ strongly at all time, separating agents into $n$ clusters.
    \end{proof}
    
    \begin{remark} \label{rem2}
    So far we have obtained necessary conditions and a sufficient condition based on the assumption that the unsuccessful recovery, i.e., $\ed\setminus\ea$, still consumes energy as formulated in (\ref{en.d}). The conditions in the case where the defender does not lose energy from unsuccessful recovery can be obtained in more intuitive forms. In such a case, we assume that the energy consumption of the defender satisfies $\sum_{m=0}^{k} \bd|\edm\cap\eam| \leq \kd + \rd k$. Then, necessary conditions to prevent consensus are $\ra/\bab \geq \lambda$ if $b=0$ and $\ra/\ba \geq \lambda$ otherwise. A sufficient condition can be easily obtained as $\ra/\bab \geq |\eo|$.
    \end{remark}
    
    The terms used for the necessary condition and the sufficient condition in Remark~\ref{rem2}, i.e., $\ra/\bab \geq \lambda$, $\ra/\ba \geq \lambda$, and $\ra/\bab \geq |\eo|$, are the same as those in Propositions~\ref{lem3.4}~and~\ref{lem3.5} above, since the conditions of those results are derived from the attacker's ability rather than the defender's, as discussed in the proofs of the propositions. For the case of $b=0$, there is a difference from Proposition~\ref{lem3} that in a no-waste energy situation, the defender's horizon parameters no longer influence the requirement for obtaining a tighter necessary condition. This implies that the defender becomes weaker with energy constraint $(\ref{en.d})$, since the necessary conditions to prevent consensus become less tight.
    
    {\color{black}\begin{remark}
    For the node attack case characterized in Remark~\ref{rem1}, the necessary conditions $\ra/\bab\geq \lambda$ and $\ra/\ba\geq \lambda$ in Propositions~\ref{lem3.4} and \ref{lem3} to prevent consensus change to $\ra/{\overline{\beta}^{\ar}_{\vo}}\geq 1$ and $\ra/{\beta^{\ar}_{{\vo}}}\geq 1$, respectively, since the attacker only needs to isolate an agent to prevent consensus. 
    \end{remark}}

    \vspace{-0.1cm}
    \section{Clustering Analysis}\label{thrf}
    
    In this section, we derive some results on the number of formed clusters of agents at infinite time. From Proposition~\ref{lem3.5} above, it is clear that if $\ra/\bab\geq|\eo|$, then the attacker can make $n$ clusters by strongly attacking all edges at all time. Thus, for the result below, we consider the case where $\ra/\bab<|\eo|$.
    \begin{proposition}\label{lem6}
        Define a vector $\Theta\in \mathbb{R}^{|\eo|}$ with elements $\Theta_i:=\max_{|\eo^{\ar}|=i} \overline{n}((\vo,{\eo}\setminus{\eo^{\ar}}))$, with $\overline{n}((\vo,{\eo}\setminus{\eo^{\ar}}))$ being the number of groups of graph $(\vo,{\eo}\setminus{\eo^{\ar}})$. Then the number of formed clusters at infinite time is upper bounded by
        \begin{itemize}
            \item $\Theta_{\floor{\ra/\bab}}$ if either
            \begin{enumerate}[(a)]
                \item $b=0$, $\hd\geq\ha$ and $\mathrm{lcm}(\ta,\td)=\ta$; or
                \item $b=0$ and $\td=1$,
            \end{enumerate}
            \item $\Theta_{\min\{|\eo|,\floor{\ra/\ba}\}}$ otherwise.
        \end{itemize}
    \end{proposition}
    \vspace{5pt}
    \begin{proof}
    The $i$-th element of the vector $\Theta$ consists of the maximum number of formed groups $\overline{n}((\vo,{\eo}\setminus{\eo^{\ar}}))$ given the number of attacked edges to be $i$. As some edges need to be attacked consistently to divide the agents into clusters, the number of formed clusters at infinite time is never more than the maximum number of groups at any time $k$ given the same number of strongly attacked edges.
    
    The rest of the proof follows from Proposition~\ref{lem3}. For the cases specified in $(a)$ and $(b)$ there, $\floor{\ra/\bab}$ is the maximum achievable number of edges that can be strongly attacked at all times. If the conditions of Cases~(a) or (b) do not hold, then the maximum number of edges is $\floor{\ra/\ba}$. Thus, given the known graph topology $\go$ with $|\eo|$ number of edges, we can imply that depending on the values of horizon lengths and the game periods $\ha$, $\hd$, $\ta$, and $\td$, the values in $\Theta_{\floor{\ra/\bab}}$ or in $\Theta_{\min\{|\eo|,\floor{\ra/\ba}\}}$ give the maximum number of clusters at infinite time.
    \end{proof}

    We illustrate the results in Proposition~\ref{lem6} by considering the unattacked version of the graph in Fig.~\ref{nod} as $\go$. In this graph, the vector $\Theta$ is $\Theta=[2,2,3,4]^{\mathrm{T}}$. Suppose $\ba=1$, $\bab=2$, $\ra=3.5$. In this case, if $b=0$ and $\ha,\ \hd, \ \ta, \td$ satisfy conditions (a) or (b) in Proposition~\ref{lem6}, then the maximum number of clusters is $\Theta_1=2$; otherwise the maximum number of clusters is $\Theta_3=3$. When the values of horizon lengths and game periods satisfy condition (a) or (b) in Proposition~\ref{lem6}, it follows that $\Theta_{i+1}\geq \Theta_i$ holds for each $i$. This is because under condition (a) or (b), the defender is stronger and thus,
the attacker may not be able to make more clusters when the number of attacks increase.
    
    {\color{black}\begin{remark}
    For the node attack case characterized in Remark~\ref{rem1}, in the upper bounds of the clusters in Proposition~\ref{lem6}, the vector $\Theta$ changes to ${\Theta}_{\vo}\in \mathbb{R}^{|\vo|}$ with elements ${\Theta}_{\vo,i}:=\max_{|\vo^{
    \ar}|=i}\overline{n}_{\vo}((\vo\setminus\vo^{\ar},{\eo}\setminus{\eo_{\vo}^{\ar}}))$, where $\overline{n}_{\vo}((\vo\setminus{\vo^{\ar}},{\eo}\setminus{\eo_{\vo}^{\ar}}))$ is the number of groups of graph $(\vo\setminus\vo^{\ar},{\eo}\setminus{\eo_{\vo}^{\ar}})$ and $\eo_{\vo}^{\ar}$ is the set of all edges adjacent to agents $\vo^{\ar}$. Assuming $\ra/\overline{\beta}^{\ar}_{\vo}<|\vo|$, the terms $\Theta_{\floor{\ra/\bab}}$ and $\Theta_{\floor{\ra/\ba}}$ then change to $\Theta_{\vo,\floor{\ra/\overline{\beta}^{\ar}_{\vo}}}$ and $\Theta_{\vo,\floor{\ra/\beta^{\ar}_{\vo}}}$, respectively.
    \end{remark}}
    
\section{Numerical Example}
    \begin{figure}
        \includegraphics[trim=8 0 0 0,clip, scale=0.6]{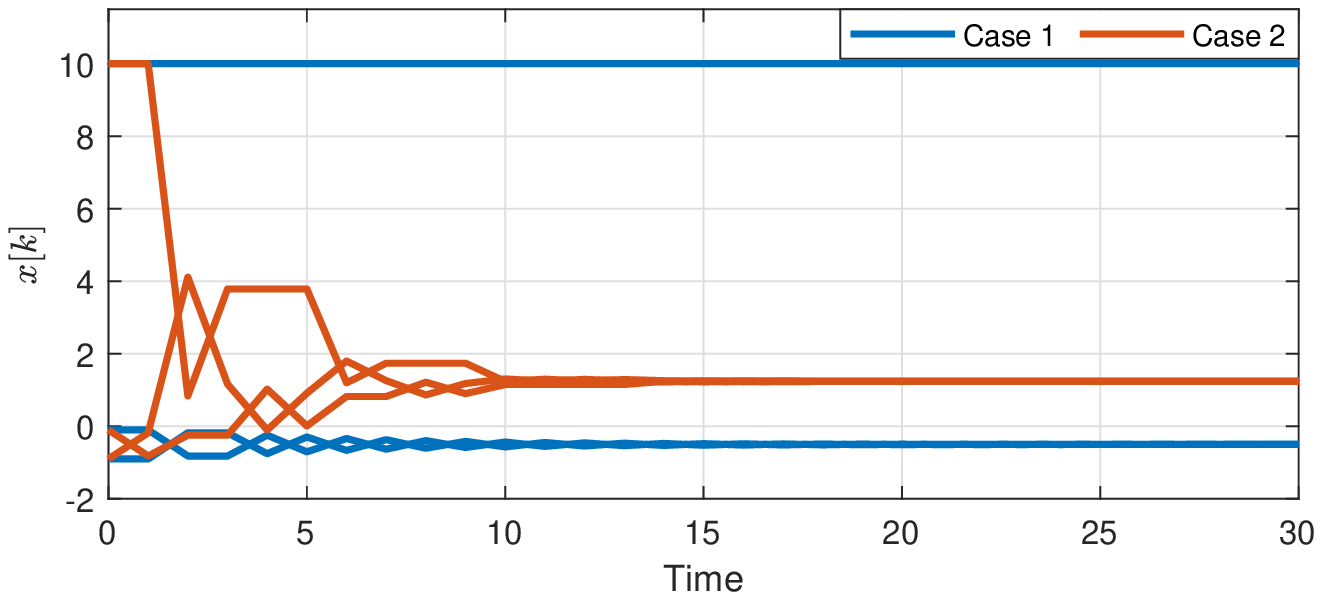}
         \vskip -7pt
        \caption{Agents' states by the defender for Cases~1~and~2}
        \label{fig3}
        \vspace{5pt}
        \includegraphics[trim=8 0 0 0,clip, scale=0.6]{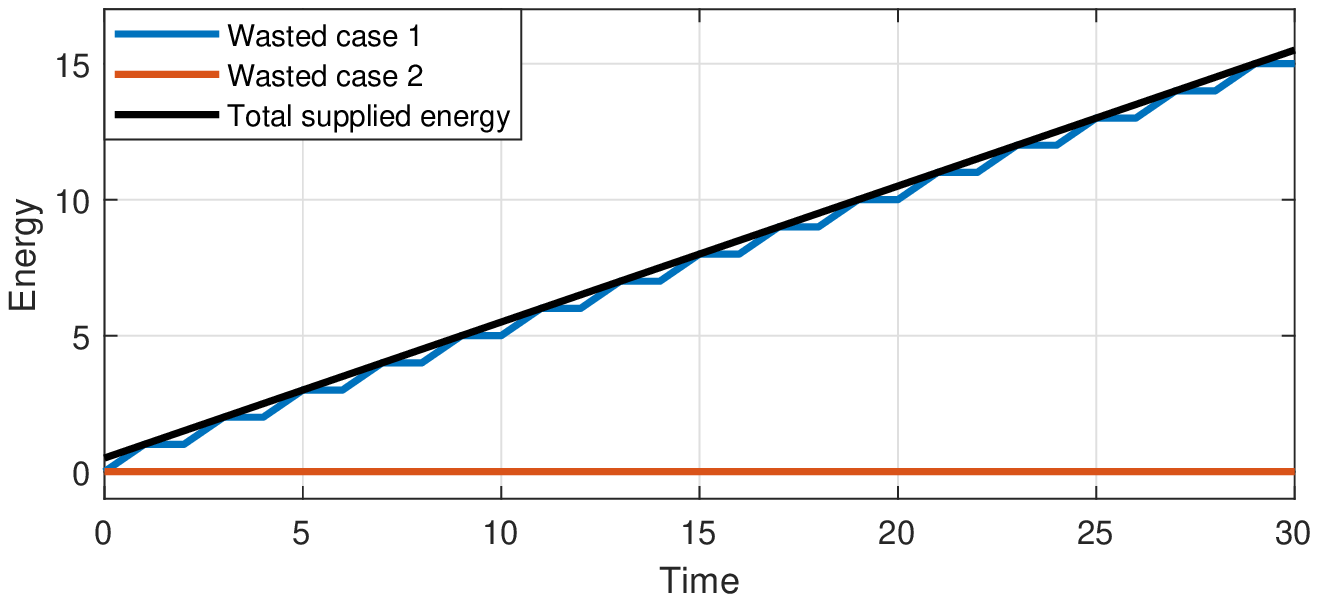}
         \vskip -7pt
        \caption{Wasted energy and total supplied energy, i.e., $\kd+\rd k$, of the defender for Cases~1~and~2}
        \label{fig2}
    \end{figure}

In this section, we show the effects of the horizon lengths and game periods on the consensus speed and the number of clusters by numerical simulation. Specifically, we demonstrate that consensus may still be prevented if the defender's horizon parameters are sufficiently small, even with relatively low energy parameters for the attacker. This is related to the conditions in Propositions~\ref{lem3.4}~and~\ref{lem3}.

Consider a simple path graph consisting of three agents ($|\eo|=2, \lambda=1$) with the following parameters: 
\begin{itemize}
    \item Case 1: $\ha=3$, $\hd=2$, $\ta=1$, $\td=2$,
    \item Case 2: $\ha=\hd=\ta=\td=2$,
\end{itemize}
with $b=0$, $\kd=\rd=0.5$, $\ka=\ra=1.5$, $\bab=2$, $\ba=\bd=1$ in both cases. Notice that since $\ra/\bab<\lambda=1$, the attacker is not able to strongly attack edges at all time to keep the graph disconnected. Thus, in order to prevent consensus, the attacker needs to continuously change its strategies to make the recovery unsuccessful. 

The attacker is stronger than the defender in Case~1, since with $\ha>\hd$ and $\ta<\td$ the attacker can look further forward and update their strategies more often. Consequently, in this case the attacker may avoid the recovery on $\ed$ either by canceling its planned attacks or by changing to strong attacks instead. In Case~2, both players have the exact same horizon parameters, implying that the defender never wastes its energy since the attacker is not able to unilaterally change its strategy.

Fig.~\ref{fig3} shows the evolution of agents' states where agents are being divided into two clusters in Case 1 while in contrast, they converge to the same state in Case~2. This is because in Case~2 the defender wastes all of its its energy by attempting to allocate its resources to the edges that are not attacked normally, as illustrated in Fig.~\ref{fig2}. Note that the values of the horizon lengths and the game periods in Case~2 satisfy the requirements in Proposition~\ref{lem3} to make the necessary condition tighter, i.e., $\ra/\bab\geq\lambda$ instead of of $\ra/\ba\geq\lambda$. On the other hand, the horizon parameters in Case~1 do not satisfy those requirements, making it easier to prevent consensus.

\vspace{-0.1cm}
\section{Conclusion}
We have formulated a two-player game in a cluster formation of resilient multiagent systems. The players consider the impact of their actions on future communication topologies and agent states, and adjust their strategies according to a rolling horizon approach. Conditions for forming clusters among agents have been derived. We have discussed the effect of the horizon parameters on the possible number of clusters and consensus. In general, the attacker needs to have sufficiently long horizon length and short game period to prevent consensus, in addition to having sufficient energy for generating attacks.
\nocite{*}
\bibliographystyle{IEEEtran}
\bibliography{IEEEabrv,arxedit}

\end{document}